\documentclass[sigconf]{acmart}

\usepackage{tikz}
\usepackage{macros}
\AtBeginDocument{%
  }

\setcopyright{acmlicensed}
\copyrightyear{2026}
\acmYear{2026}
\acmDOI{XXXXXXX.XXXXXXX}

\usepackage{todonotes}

\usepackage{amsthm}
\newtheorem{theorem}{Theorem}[section]
\newtheorem{lemma}{Lemma}[section]

\theoremstyle{definition}
\newtheorem{definition}{Definition}[section]
\newtheorem{example}{Example}[section]

\newtheorem{remark}{Remark}[section]

\usetikzlibrary{automata, arrows.meta, positioning}

\begin{document}
\title{Runtime Semantic Discovery for Evolving Action Alphabets}

\author{Pietro Bernabei}
\authornote{Both authors contributed equally to this research.}
\email{pietro.bernabei@imtlucca.it}
\email{pietro.bernabei@unifi.it}
\orcid{0009-0004-0582-4591}
\affiliation{%
  \institution{IMT School for Advanced Studies Lucca}
  \city{Lucca}
  \country{Italy}
}

\author{Rosario Pugliese}
\authornotemark[1]
\email{rosario.pugliese@unifi.it}
\orcid{0000-0002-1419-1405}
\affiliation{%
  \institution{University of Florence}
  \city{Florence}
  \country{Italy}
}


\begin{abstract}
Runtime Verification typically assumes a fixed alphabet of observable actions. This assumption becomes fragile in evolving software systems, where updates, reconfigurations, or architectural changes may rename, split, or fuse actions while preserving the system’s high-level intent. In these cases, previously valid monitors and specifications may become misaligned with the running system.

In this paper, we study runtime semantic discovery, namely the problem of identifying whether newly observed actions play the same operational role as previously known ones, without reconstructing a global behavioural model. To this end, we propose a local and runtime-computable notion of action similarity based on the neighbour contexts in which actions occur within execution traces. We define a family of k-neighbour equivalences and use them to compare systems whose observable alphabets differ.

Our approach is intentionally lightweight: it does not aim to fully characterise behavioural preservation, but rather to identify candidate semantic correspondences between actions in evolving systems. We discuss the expressive power and limitations of neighbour-based equivalence and show how inferred correspondences can support the adaptation of runtime monitors and specifications under selected classes of alphabet evolution.
\end{abstract}


\renewcommand{\shortauthors}{Bernabei and Pugliese}



\maketitle
\section{Introduction}
\label{sec:intro}

Modern software systems evolve continuously \cite{islam_runtime_2023}. Updates, reconfigurations, service replacement, and architectural refactorings may change the set of observable actions produced at runtime, even when the system preserves its overall functional intent. For Runtime Verification (RV), this creates a significant challenge: monitors and specifications are usually defined over a fixed alphabet, which may therefore become obsolete as the monitored system evolves, even when the system’s functional behaviour remains conceptually unchanged.

This problem is particularly relevant in settings where behavioural adaptation happens frequently and with limited human supervision, as in the Continuous Integration/Continuous Development process \cite{cerny_analyzing_2025}. In such scenarios, manually revising specifications or rebuilding behavioural models from scratch is often too expensive or too slow. What is needed instead is a lightweight runtime mechanism able to relate newly observed actions to previously known ones, so that existing monitoring artefacts can remain useful across system evolution.

In this paper, we investigate whether this relation can be inferred directly from the local structure of execution traces. Our starting intuition is simple: actions that play the same operational role often occur in similar local contexts, even if their labels or implementations change. This suggests using trace neighbours as a basis for identifying semantic continuity across evolving alphabets.

To formalise this idea, we introduce a notion of neighbour context for action occurrences and define a hierarchy of k-neighbour equivalences. These equivalences provide a local and runtime-computable criterion for comparing actions and systems with partially different observable alphabets. The goal is not to replace stronger behavioural relations such as trace equivalence or bisimulation, but rather to provide a practical structural invariant for detecting reasonable correspondences between old and new actions.

System evolution is fundamentally driven by corrective, adaptive, perfective, or extensive maintenance needs, a taxonomy frequently adopted in software engineering (e.g., \cite{ketfi_dynamic_2002}). At the level of observable behaviour, these high-level adaptations often manifest as structural changes to the system's alphabet. Crucially, in addressing these scenarios, the aim is not to recover full semantic equivalence, but rather to derive local evidence for candidate action correspondences. This lightweight perspective is particularly suited to capturing these concrete manifestations, focusing on common evolutionary patterns such as the renaming, splitting, and fusing of actions. In such cases, the monitor is not simply facing a different trace, but a changed observational vocabulary. Our proposal is to use neighbour preservation as evidence that some actions may still represent the same functional role from the viewpoint of monitoring.

The contribution of the paper is threefold:

\begin{enumerate}
    \item we define a local notion of action similarity based on neighbour contexts in execution traces, together with its generalisation to k-neighbour equivalence;
    \item we introduce an abstraction-based framework for comparing systems whose observable alphabets differ, capturing typical alphabet evolutions such as renaming, splitting, fusing, and more general many-to-many mappings;
    \item we clarify both the usefulness and the limitations of neighbour-based equivalence, arguing that it should be interpreted as a mechanism for deriving candidate semantic mappings that can support monitor and specification adaptation, rather than as a complete characterisation of behavioural preservation.
\end{enumerate}

Overall, the paper aims to isolate a tractable fragment of the broader problem of specification evolution in Runtime Verification, focusing on what can be inferred incrementally and online from the observed execution context.

The rest of the paper is organised as follows.

Sec.~\ref{sec:motivating_example} introduces a simple motivating example.
Sec.~\ref{sec:relwork} reviews the most closely related work.
Sec.~\ref{sec:methodology} presents the methodological perspective of the paper.
Sec.~\ref{sec:neighbour-equiv} develops the formal framework based on neighbour and $k$-neighbour equivalence, together with an abstraction mechanism for comparing systems with partially different alphabets.
Sec.~\ref{sec:spec-updat} discusses monitor and specification adaptation under alphabet evolution.
Sec.~\ref{sec:conclusion} concludes and outlines future research directions.

\section{Motivating Example}
\label{sec:motivating_example}

Consider a Secure Information Retrieval (SIR) system similar to the one in Figure~\ref{fig:SIR}. At a high level, the system includes a request processor, a certification component, a data sender, a logger, and an encryption component. Runtime verification is used to monitor whether messages are properly processed before being transmitted. In particular, one may want to check that every message is encrypted before it is sent over the network.

Assume that the original system emits an observable action named 'Encryptor' whenever the encryption step is performed.  We could therefore be interested in monitoring a possible requirement of 'send only encrypted messages'. For instance, this requirement can be specified in LTL as $\varphi=\Box(\neg send\text \ U \  \mathit{Encryptor})$. 
Now, suppose that, due to a runtime update, the encryption component is replaced by a different implementation. In the simplest case, the system may expose a new action, such as 'CryptoEngine', instead of 'Encryptor'. While the workflow's functional intent remains unchanged, the monitor now observes an action alphabet different from the one used when the specification was defined. Consequently, the original monitor may fail to recognise CryptoEngine as the action fulfilling the role previously associated with Encryptor. So to maintain the continuity of the monitoring, $\varphi$ must be updated to $\varphi'=\Box(\neg send \ U \ CryptoEngine)$.

This example illustrates the central difficulty addressed in this paper: system evolution may preserve the operational role of a step while changing the observable actions through which that step is represented. In simple cases, this change resembles a renaming, but in general, it may involve finer-grained refinements, as many-to-one and one-to-many correspondences between old and new actions. In all such cases, the key question is whether a runtime monitor can identify which newly observed actions play the same role as previously known ones, without reconstructing a full behavioural model of the updated system.

Our working intuition is that the role of an action can often be inferred from the local context in which it appears. If a newly observed action systematically occurs between the same surrounding phases of the workflow as a previously known action, then it is a plausible candidate to represent the same operational role from the monitor’s perspective. This observation motivates the neighbour-based approach developed in the following sections, where we compare actions based on the local contexts induced by execution traces rather than on global behavioural equivalence.

\begin{figure*}
\centering
\includegraphics[width=1\linewidth]{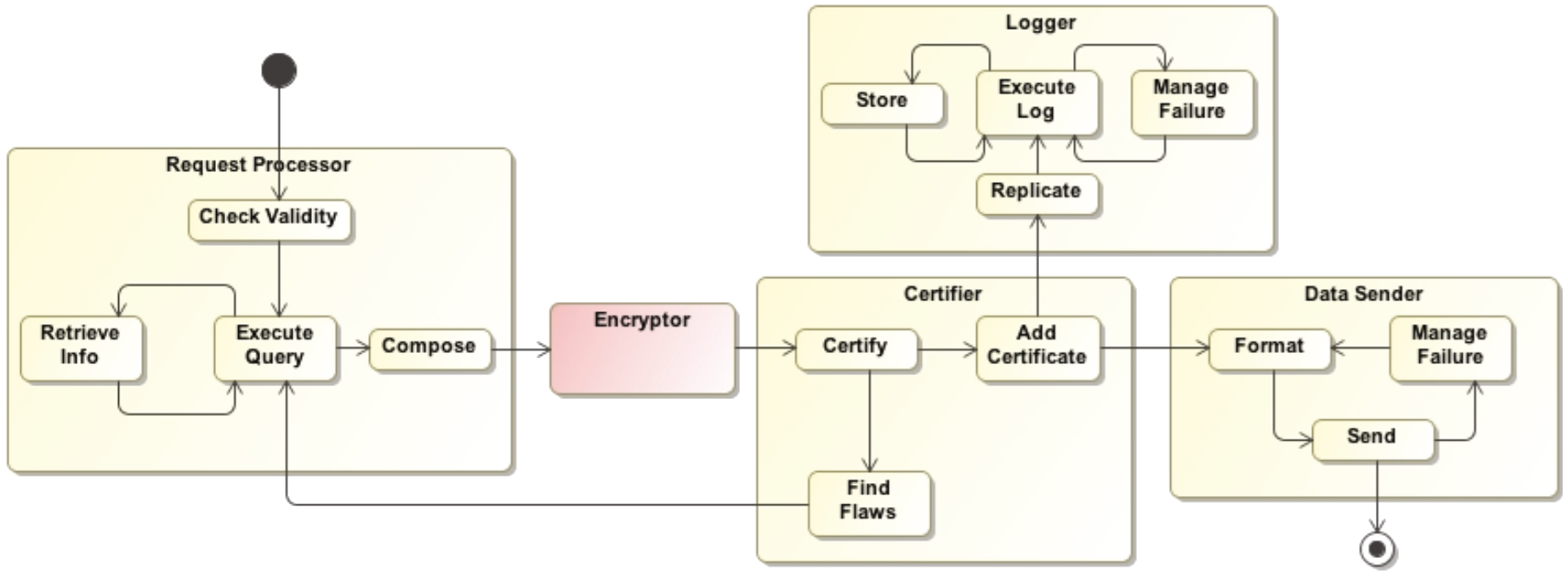}
    \caption{Secure Information Retrieval workflow}
    \label{fig:SIR}
    \Description{A workflow scheme of Secure Information Retrieval. The SIR workflow begins in the top-left corner with the Request Processor, which validates the request and initiates an iterative cycle of querying and retrieving information until the data is compiled. The package then passes through the Encryptor, a central security module that encrypts the output before sending it to the Certifier. At this stage, a quality control check verifies for anomalies: if any are detected, the system activates a feedback loop that returns the flow to the initial query; if successful, a digital certificate is applied. This operation splits the path: upwards, it activates the Logger module, which manages the parallel replication and persistent archiving of the activity, whilst to the right it proceeds to the Data Sender. Here, the data is formatted and transmitted; if transmission fails, a local recovery mechanism returns the packet to the formatting stage for a new attempt. The process formally concludes at the bottom-right, ensuring complete traceability and data integrity.}
\end{figure*}
\section{Related Work}
\label{sec:relwork}

The problem addressed in this paper lies at the intersection of several research areas, including incremental verification, runtime verification under system variation, behavioural equivalences, action refinement, and model learning from execution data. These areas provide important background, but none of them directly addresses the question considered here: how to infer, at runtime and from local trace information only, whether newly observed actions play the same monitoring-relevant role as previously known ones.

\subsection{Incremental and adaptive verification}

A first line of related work concerns incremental verification of evolving systems. As discussed in the literature \cite{ghezzi_evolution_2012}, incremental approaches can exploit modular encapsulation \cite{seco_robust_2020}, change anticipation \cite{hutchison_lover_2013}, or syntax-driven comparison \cite{bianculli_syntactic-semantic_2013, bianculli_syntactic-semantic_2015} to avoid re-verifying a system from scratch after modification. In all these cases, the common goal is to reuse previous verification efforts when the system changes.

Our work shares the general objective of avoiding full re-analysis after evolution, but differs in both scope and method. We do not assume access to source-level changes, explicit change sets, or a known boundary separating the modified fragment from the unmodified one. Instead, we focus on the runtime setting, where the only available evidence may be the observed action stream. For this reason, our approach does not compare program syntax or global models, but relies on local structural regularities in execution traces.

\subsection{Runtime verification under variation}

A second relevant area is runtime verification for systems that evolve over time. Traditional RV assumes a fixed specification and a fixed alphabet of observable actions \cite{bartocci_lectures_2018, falcone_tutorial_2013}. This assumption is often inadequate when systems undergo reconfiguration, component replacement, or changes to the observable interface.

Some approaches extend RV to richer action domains, for example, by using register automata to deal with dynamically generated identities and data values \cite{grigore_runtime_2015}. These techniques increase the expressive power of monitoring, but they still assume that the monitoring structure is known in advance. They do not address the problem considered here, namely, inferring correspondences between old and new action labels when the observable alphabet itself evolves.

Other approaches to adaptive RV update monitors based on externally provided knowledge, such as human-defined update patterns or explicit specification changes \cite{carwehl_runtime_2023}. These methods are relevant because they acknowledge that monitors may need to evolve together with the system. However, they typically assume that the transformation to apply is already known. By contrast, our interest lies in the preceding step: deriving candidate correspondences between actions directly from runtime observations, to support monitor or specification adaptation when no explicit mapping is provided a priori.

\subsection{Behavioural equivalences and action refinement}

A classical way to compare systems is through behavioural equivalences, such as trace equivalence or bisimulation \cite{milner_communication_2006,glabbeek_linear_1990}. These notions are foundational in concurrency theory and provide robust criteria for establishing whether two systems exhibit the same observable behaviour.

However, such equivalences are not well-suited to the problem studied in this paper. First, they are inherently global: they compare systems based on their overall behaviour rather than on partial, incrementally available observations. Second, they presuppose a stable interpretation of action labels, whereas our setting is precisely one in which labels may change, 
or be replaced. For us, the goal is not to prove that two systems are behaviourally identical, but to identify whether some actions in the evolved system are plausible candidates to preserve the operational role of previously monitored ones.

A related but distinct body of work is action refinement, where an abstract action is replaced by a more detailed behaviour while preserving correctness properties \cite{aceto_towards_1993}. This perspective is usually top-down and design-time oriented: the refinement relation is part of the system development process and its direction is known in advance. Our setting is different. We observe the effects of evolution a posteriori, at runtime, and seek to reconstruct candidate correspondences between actions from their observed trace contexts rather than from a predefined refinement relation.

\subsection{Process mining and evolution-aware monitoring}

Another natural point of comparison is process mining and, more generally, model learning from execution logs \cite{van_der_aalst_process_2016, germiniani_systematic_2025}. These approaches aim to reconstruct behavioural models from traces and are highly effective when sufficient data and offline analysis are available.

Our approach is deliberately lighter. We do not attempt to learn a full process model or to recover the complete structure of the updated system. Instead, we restrict attention to local contextual information that can be extracted incrementally from traces and used online. This sacrifices expressive power, but makes the method more compatible with runtime verification constraints.

Finally, some works study the monitoring of systems across version changes by combining obligations inherited from a previous version with the specification of the new version \cite{zhang_using_2006, nahabedian_assured_2016, chait-roth_consistent_2026}. These approaches are complementary to ours: they assume that a suitable specification for the evolved system is available, and focus on how to manage the transition. In contrast, we address an earlier problem: how to recover a plausible correspondence between the old and the new action vocabularies when such a specification is not yet available.

Overall, the existing literature provides techniques for reusing verification effort, updating monitors with external guidance, comparing complete behaviours, refining actions, or learning models from logs. Our contribution is orthogonal to these directions: we study a local and runtime-computable criterion for relating actions across alphabet evolution, with the specific goal of supporting monitoring continuity in evolving systems.
\section{Methodology}
\label{sec:methodology}

Our goal is to support runtime monitoring when the observable alphabet of a system evolves across versions. We consider systems whose implementation or internal structure may change while preserving, at least in part, their high-level functional intent. In these situations, a monitor built for an earlier version may encounter previously unseen actions replacing known ones, or may observe a different decomposition of the same operational step. As a result, the monitor may lose alignment with the vocabulary used by the original specification.

The methodological assumption of this work is that such changes can often be analysed through the local context in which actions occur. Rather than comparing complete behaviours or reconstructing a global model of the updated system, we focus on the actions that precede and follow each observed action in execution traces. The intuition is that if an action in the updated system consistently appears in contexts similar to those of a previously known action, then it is a plausible candidate to preserve the same monitoring-relevant role.

This criterion is intentionally local and lightweight. It is not meant to establish full behavioural equivalence between system versions, nor to prove that two actions have the same semantics in a strong sense. Instead, it provides evidence for deriving candidate semantic correspondences between actions across the evolution of the alphabet, with the aim of supporting the adaptation of runtime monitors and specifications. The advantage of this choice is that the analysis can be carried out incrementally, directly on observed traces, without assuming explicit knowledge of the transformation that generated the updated system.


To make this setting precise, we consider a small set of recurring classes of correspondences induced by alphabet evolution. These are not assumed to be known at runtime. Instead, they serve as conceptual classes of evolution patterns. In the simplest case, evolution appears as \textbf{renaming}, where one observable action is replaced by another. More generally, an action may be replaced by a set of alternative actions (\textbf{splitting}), or multiple actions may be abstracted into a single one (\textbf{fusing}).  Furthermore, an atomic action might be refined into a concrete \textbf{sequence} of operations, or we may encounter more general \textbf{multi-mapping} situations, where the relation between old and new observable actions is neither one-to-one nor known in advance. 
Figure~\ref{fig:evolution_system} shows representative forward evolution patterns; fusing is discussed as the inverse abstraction of splitting rather than as a separate forward transformation.

These classes do not define the correspondence on their own; rather, they identify the kinds of evolutionary scenarios that the proposed framework is intended to analyse. In the next section, we formalise this idea by introducing neighbour contexts over traces, extending them to $k$-neighbour contexts, and using them to compare systems with equal or partially different observable alphabets.

\begin{figure}
\centering
\begin{tikzpicture}[>=stealth, auto, thick, node distance=1cm] 
\tikzstyle{state}=[circle, draw]

\begin{scope}[shift={(0,0)}]
    \node[draw=none] at (0,0.7) {\large $S_0$};
    \node[state] (q1) {q1};
    \node[state, below=of q1] (q2) {q2};
    \node[state, below left= of q2] (q3) {q3};
    \node[state, below right=of q2] (q4) {q4};
    \node[state, below=of q3] (q5) {q5};
    \node[state, below=of q4] (q6) {q6};
    \draw[->] (q1) -- node[left]{$b$} (q2);
    \draw[->] (q2) -- node[left]{$a$} (q3);
    \draw[->] (q2) -- node[left]{$a$} (q4);
    \draw[->] (q3) -- node[left]{$c$} (q5);
    \draw[->] (q4) -- node[left]{$d$} (q6);
\end{scope}

\begin{scope}[shift={(5,0)}]
    \node[draw=none] at (0,0.7) {\large $S_1$};
    \node[state] (q1b) {q1}; 
    \node[state, below=of q1b] (q2b) {q2};
    \node[state, below left= of q2b] (q3b) {q3};
    \node[state, below right=of q2b] (q4b) {q4};
    \node[state, below=of q3b] (q5b) {q5};
    \node[state, below=of q4b] (q6b) {q6};
    \draw[->] (q1b) -- node[left]{$b$} (q2b);
    \draw[->] (q2b) -- node[left]{$f$} (q3b);
    \draw[->] (q2b) -- node[left]{$f$} (q4b);
    \draw[->] (q3b) -- node[left]{$c$} (q5b);
    \draw[->] (q4b) -- node[left]{$d$} (q6b);
\end{scope}

\begin{scope}[shift={(0,-6.2)}]
   \node[draw=none] at (0,0.7) {\large $S_2$};
    \node[state] (q1c) {q1};
    \node[state, below=of q1c] (q2c) {q2};
    \node[state, below left= of q2c] (q3c) {q3};
    \node[state, below right=of q2c] (q4c) {q4};
    \node[state, below=of q3c] (q5c) {q5};
    \node[state, below=of q4c] (q6c) {q6};
    \draw[->] (q1c) -- node[left]{$b$} (q2c);
    \draw[->] (q2c) -- node[left]{$f$} (q3c);
    \draw[->] (q2c) -- node[left]{$e$} (q4c);
    \draw[->] (q3c) -- node[left]{$c$} (q5c);
    \draw[->] (q4c) -- node[left]{$d$} (q6c);
\end{scope}

\begin{scope}[shift={(5,-6.2)}]
    \node[draw=none] at (0,0.7) {\large $S_3$};
    \node[state] (q1f) {q1}; 
    \node[state, below=of q1f] (q2f) {q2};
    \node[state, below left= of q2f] (q3f) {q3};
    \node[state, below right=of q2f] (q4f) {q4};
    \node[state, below=of q3f] (q5f) {q5};
    \node[state, below=of q4f] (q6f) {q6};
    \node[state, below=of q5f] (q7f) {q7};
    \draw[->] (q1f) -- node[left]{$b$} (q2f);
    \draw[->] (q2f) -- node[left]{$a$} (q3f);
    \draw[->] (q2f) -- node[left]{$a$} (q4f);
    \draw[->] (q3f) -- node[left]{$e$} (q5f);
    \draw[->] (q4f) -- node[left]{$d$} (q6f);
    \draw[->] (q5f) -- node[left]{$f$} (q7f);
\end{scope}
\begin{scope}[shift={(0,-12.4)}]
   \node[draw=none] at (0,0.7) {\large $S_4$};
    \node[state] (q1e) {q1};
    \node[state, below=of q1e] (q2e) {q2};
    \node[state, below left= of q2e] (q3e) {q3};
    \node[state, below right=of q2e] (q4e) {q4};
    \node[state, below=of q3e] (q5e) {q5};
    \node[state, below=of q4e] (q6e) {q6};
    \draw[->] (q1e) -- node[left]{$b$} (q2e);
    \draw[->] (q2e) -- node[left]{$a$} (q3e);
    \draw[->] (q2e) -- node[left]{$a$} (q4e);
    \draw[->] (q3e) -- node[left]{$e$} (q5e);
    \draw[->] (q4e) -- node[left]{$f$} (q6e);
\end{scope}

\end{tikzpicture}
\caption{Evolution of the initial system $S_0$ through representative alphabet-evolution patterns: $S_1$ (renaming), $S_2$ (splitting), $S_3$ (sequence refinement), and $S_4$ (multiple mapping).
Fusing is treated as the inverse abstraction of splitting, i.e., collapsing the unshared actions of $S_2$ back into the corresponding action of $S_0$.}
\label{fig:evolution_system}
\end{figure}
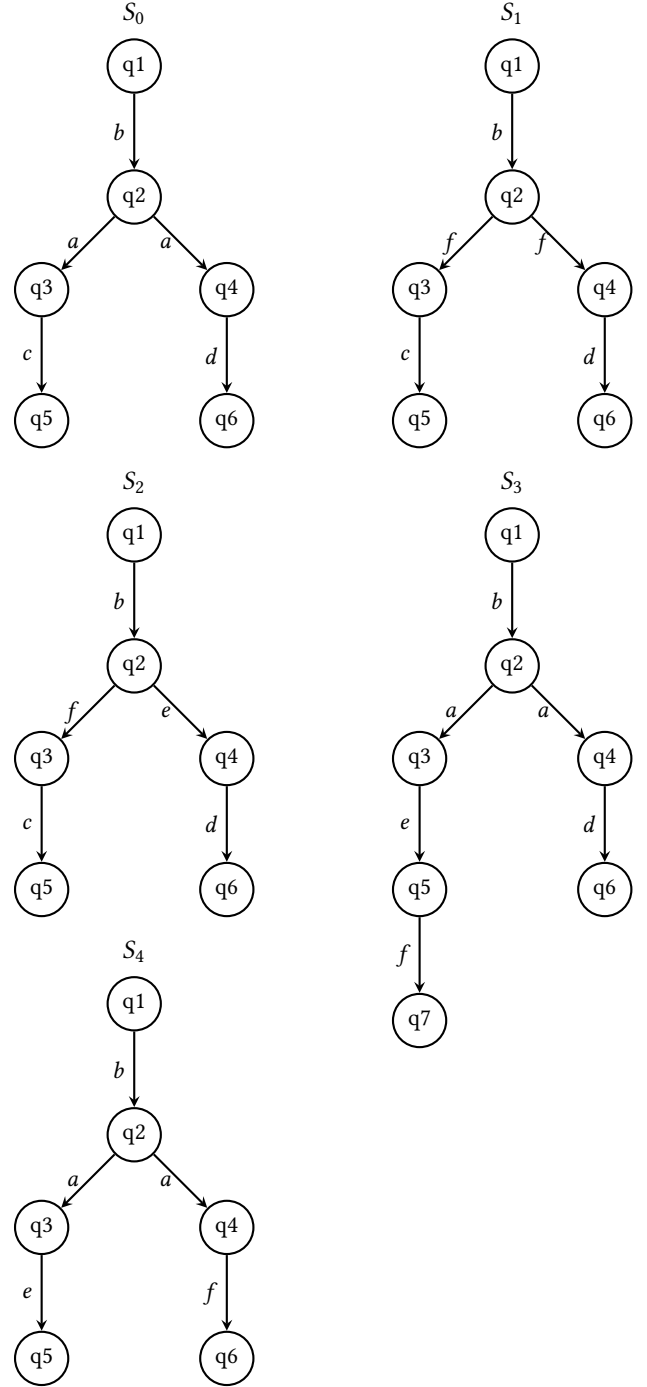

\section{Neighbour Equivalence}
\label{sec:neighbour-equiv}

In this section, we formalise the local criterion introduced above. We first define neighbour contexts over traces and use them to compare systems with the same observable alphabet. We then generalise the construction to $k$-neighbours and introduce an abstraction mechanism for comparing systems with different alphabets. The purpose of this framework is not to establish full behavioural equivalence, but to derive structural evidence for \emph{candidate semantic correspondences} between actions observed across system evolution.

Let $\SysS = (Q, A, \to, q_0)$ be a Labelled Transition System (LTS), where $Q$ is a set of states, $A$ is a finite set of observable actions, $\to \subseteq Q \times A \times Q$ is the transition relation, and $q_0 \in Q$ is the initial state. We write $q \xrightarrow{a} q'$ whenever $(q,a,q') \in \to$.

In this paper, we consider \emph{finite traces} generated from the initial state. A \textit{trace} of $\SysS$ is a finite word $t = a_1 a_2 \dots a_n \in A^*$ such that there exists a corresponding path in the LTS starting from $q_0$. Let $t = a_1 a_2 \dots a_n$ be a trace. A sequence $t' \in A^*$ is a \emph{subtrace} of $t$, written $t' \sqsubseteq t$, if
there exist indices $i,j$ such that $1 \leq i \leq j \leq n$ and
\[
t' = a_i a_{i+1} \dots a_j.
\]
That is, $t'$ is a contiguous subsequence of $t$.

We denote by $\Lang{\SysS}$ the set of finite traces of $\SysS$. 
Two systems $\SysA$ and $\SysB$ over the same alphabet $A$ are deemed to be \emph{trace equivalent} if and only if $\Lang{\SysA} = \Lang{\SysB}$.

The basic intuition of our approach is that an occurrence of an action can be characterised by the actions that immediately precede and follow it in a trace.

Let $t = a_1 a_2 \dots a_n \in \Lang{\SysS}$. For each position $i \in \{1,\dots,n\}$, we define the \emph{neighbour} of the occurrence $a_i$ in $t$ as
\[
\PairTraceExt{t}{a_i} \triangleq (a_{i-1}, a_{i+1}),
\]
where we use a distinguished boundary symbol $\Terminator\notin A$, with $a_0 = \Terminator$ and $a_{n+1} = \Terminator$.

For an action $a \in A$, we define its neighbour set in the trace $t$ as
\[
\PairTraceExt{t}{a} \triangleq \{\, \PairTraceExt{t}{a_i} \mid a_i = a \,\}.
\]

We then lift this construction from traces to systems:
\[
\Pair{\SysS}{a} \triangleq \bigcup_{t \in \Lang{\SysS}} \PairTraceExt{t}{a},
\]
Finally, we define the neighbour representation of the system as
\[
    \PairSys{\SysS} \triangleq \{\, (a, \Pair{\SysS}{a}) \mid a \in A \,\}.
\]
This representation explicitly keeps the association between each action and the set of local contexts in which it occurs. This is significant because the collapse of all neighbours into a single unlabelled union would result in the loss of the identity of the central action. 


We can now define a first local equivalence between systems with the same observable alphabet.

\begin{definition}[Neighbour equivalence]
Let $\SysA$ and $\SysB$ be two systems over the same alphabet $A$. We say that $\SysA$ and $\SysB$ are \emph{neighbour equivalent}, written
\[
\SysA \simeq_N \SysB,
\]
if and only if
\[
\PairSys{\SysA} = \PairSys{\SysB}.
\]
\end{definition}

This relation compares systems' actions, action by action, through the local contexts in which each action occurs.

The connection with trace equivalence is one-sided.

\begin{theorem}
If two systems are trace equivalent, then they are neighbour equivalent:
\[
\Lang{\SysA} = \Lang{\SysB} \Longrightarrow \SysA \simeq_N \SysB.
\]
\end{theorem}
\begin{proof}
Assume $\Lang{\SysA} = \Lang{\SysB}$. We must show that $\PairSys{\SysA} = \PairSys{\SysB}$, which holds if and only if $\Pair{\SysA}{a} = \Pair{\SysB}{a}$ for all $a \in A$.

Let $a \in A$ and take any neighbour pair $(u,v) \in \Pair{\SysA}{a}$. By definition, there exists a trace $t \in \Lang{\SysA}$ and a position $i$ such that $a_i = a$ and $\PairTraceExt{t}{a_i}=(u,v)$. 

Since $\Lang{\SysA} = \Lang{\SysB}$, the same trace $t$ also belongs to $\Lang{\SysB}$. As neighbour pairs are defined solely from the trace, the same occurrence of $a$ in $t$ yields the same neighbour pair in $\SysB$. Hence $(u,v) \in \Pair{\SysB}{a}$, proving $\Pair{\SysA}{a} \subseteq \Pair{\SysB}{a}$.

The converse inclusion follows symmetrically. Therefore, $\Pair{\SysA}{a} = \Pair{\SysB}{a}$ for all $a \in A$, and thus $\PairSys{\SysA} = \PairSys{\SysB}$, i.e., $\SysA \simeq_N \SysB$.

\end{proof}

The converse does not hold, meaning that neighbour equivalence is weaker than trace equivalence, as the following two examples show.

\begin{example}\label{es:repeated_a}   
Consider the two single-trace systems 
\[
\Lang{\SysA}=\{baaaaac\}, \qquad \Lang{\SysB}=\{baaac\}.
\]
These systems induce the same neighbour set for action $a$, even though their traces differ. Intuitively, once a repeated block becomes long enough, the immediate-neighbour view can no longer distinguish its exact length.
\end{example}

\begin{example}\label{es:long-range}   
A second example concerns long-range dependencies across traces. Consider
\[
\Lang{\SysA}=\{axyb,\; cxyd\}, \qquad
\Lang{\SysB}=\{axyd,\; cxyb\}.
\]
Here, the neighbours' sets coincide, but the two languages are different:
\[
\PairSys{\SysA}=\PairSys{\SysB}
\qquad\text{while}\qquad
\Lang{\SysA}\neq \Lang{\SysB}.
\]
Thus, equality of neighbours does not imply equality of trace languages.
\end{example}

To increase the discriminating power of the local criterion, we generalise neighbours from immediate neighbours to windows of size $k$.

Let $t = a_1 a_2 \dots a_n$. For $k \geq 1$, the \emph{$k$-neighbour} of the occurrence $a_i$ is defined as
\[
\KPairTraceExt{t}{a_i} \triangleq (a_{i-k}, \dots, a_{i-1}, a_{i+1}, \dots, a_{i+k}),
\]
again using the boundary symbol $\Terminator$ whenever an index falls outside the trace.

As before, we define
\[
\KPairTraceExt{t}{a} \triangleq \{\, \KPairTraceExt{t}{a_i} \mid a_i = a \,\},
\qquad
\KPairSysComplete{k}{\SysS}{a} \triangleq \bigcup_{t \in \Lang{\SysS}} \KPairTraceExt{t}{a},
\]
and
\[
\KPairSysExt{k}{\SysS} \triangleq \{\, (a, \KPairSysComplete{k}{\SysS}{a}) \mid a \in A \,\}.
\]

To maintain a clear description we write $\PairSys{}$ to identify $\PairSys{1}$, and we make $k$ explicit in the other cases.

\begin{definition}[$k$-Neighbour Equivalence]
Let $\SysA$ and $\SysB$ be two systems over the same alphabet $A$. For $k \geq 1$, we say that they are \emph{$k$-neighbour equivalent}, written
\[
\SysS_1 \simeq_N^k \SysS_2 \iff \KPairSys{\SysS_1} = \KPairSys{\SysS_2}.
\]
\end{definition}

\begin{lemma}[Neighbour Projection]\label{lem:projection}
For every system $\SysS$, every $a \in A$ and every $k \geq 1$, every $k$-neighbour can be obtained as the projection of a $(k+1)$-neighbour. Formally,
\[
\KPairSysComplete{k}{\SysS}{a}
=
\{\, \pi_k(v) \mid v \in \KPairSysComplete{k+1}{\SysS}{a}\,\},
\]
where $\pi_k$ is the projection removing the outermost actions of a $(k+1)$-neighbour:

$
\qquad \pi_k(a_{i-(k+1)},\dots,a_{i-1},a_{i+1},\dots,a_{i+(k+1)})
$

$
\qquad =
(a_{i-k},\dots,a_{i-1},a_{i+1},\dots,a_{i+k}).
$
\end{lemma}
\begin{proof}
By definition, a $(k+1)$-neighbour records the $k+1$ actions immediately preceding and following a given position in a trace. Removing the outermost actions yields exactly the $k$ actions on each side, which correspond to the $k$-neighbour of the same position. Hence, every $k$-neighbour arises as the projection of a $(k+1)$-neighbour, and the equality follows.
\end{proof}
The family of $k$-neighbour equivalences forms a hierarchy: increasing $k$ increases local visibility and therefore the discriminating power of the comparison. In particular, every $(k+1)$-neighbour induces a corresponding $k$-neighbour by projection, so we have that the equivalence relations $\simeq_k$ form a hierarchy whose discriminating power increases with $k$, as stated by the following theorem.
\begin{theorem}[Hierarchy of $k$-neighbour equivalences]\label{thm:hierarchy}
For every $k \geq 1$, if $S_1 \simeq_N^{k+1} S_2$, then $S_1 \simeq_N^k S_2$.
Moreover, the implication is strict in general.
\end{theorem}

\begin{proof}
The proof consists of two parts: showing the inclusion $\simeq_{k+1} \subseteq \simeq_k$, and proving that such inclusion is strict.

\textbf{Inclusion.}
Assume $\SysS_1 \simeq_N^{k+1} \SysS_2$. By definition, their $(k+1)$-neighbour sets coincide:
\[
\KPairSysExt{k+1}{\SysS_1} = \KPairSysExt{k+1}{\SysS_2}.
\]
By applying the projection operator $\pi_k$ (Lemma~\ref{lem:projection}) for each actions to both sets we obtain:
\[
\KPairSys{\SysS_1} = \KPairSys{\SysS_2}.
\]
This implies $\SysS_1 \simeq_N^k \SysS_2$, which concludes the first part.

\textbf{Strictness.} 
To prove strictness, it suffices to show that there exist systems that are $k$-equivalent but not $(k+1)$-equivalent. 
For the base case $k=1$, Example~\ref{es:repeated_a} provides a concrete witness. 
In the example, we established that $\SysA \simeq_{\PairSys{}} \SysB$. However, inspecting the $2$-neighbours for action $a$ yields:
\[
\KPairSysExt{2}{\SysA}(a) \neq \KPairSysExt{2}{\SysB}(a).
\]
Consequently, $\SysA \not\simeq^2_{\PairSys{}} \SysB$, proving that $\simeq^2_{\PairSys{}} \subset \simeq_{\PairSys{}}$ strictly. 
The same construction can be generalised to arbitrary $k \geq 1$ by extending the repeated block, yielding systems that are $k$-neighbour equivalent but not $(k+1)$-neighbour equivalent.
Therefore, the hierarchy is strict at every level.
\end{proof}

So far, we have considered neighbour contexts for single action occurrences. However, as described previously (see Figure~\ref{fig:evolution_system}), alphabet evolution often involves replacing a single action with a sequence of actions or mapping one sequence of actions onto another. To capture these scenarios without losing the structural information of the sequence, we lift the definition of neighbour from single actions to finite sequences of actions $w \in A^+$.

\begin{definition}[Subtrace Neighbour]
Let $t = a_1 a_2 \dots a_n \in \Lang{\SysS}$ be a finite trace. For any sequence of actions $w$ of length $m \ge 1$, we define its neighbour set in the trace $t$ as the set of all boundary pairs surrounding every occurrence of $w$ in $t$:
\[
\PairTraceExt{t}{w} \triangleq \{\, (a_{i-1}, a_{i+m}) \mid 1 \le i \le n - m + 1 \text{ and } a_i \dots a_{i+m-1} = w \,\},
\]
where we assume 
$a_0 = \Terminator$ and $a_{n+1} = \Terminator$.
\end{definition}

Accordingly, we lift this construction to define the neighbour set of $w$ in the entire system $\SysS$:
\[
\Pair{\SysS}{w} \triangleq \bigcup_{t \in \Lang{\SysS}} \PairTraceExt{t}{w}
\]

The next definition characterises the set of subtraces of a system $\SysS$ formed by actions belonging to a subset $V$ of its alphabet $A$ and not further extendable by actions in $V$.

\begin{definition}[Maximal subtraces over $V$]
\label{def:variant_sequence}
Let $A$ be the observable alphabet of a system $\SysS$ and $V \subseteq A$. Given a finite trace $t \in \Lang{\SysS}$, we let 
$$
\begin{array}{@{}r@{\ }l}
V(t) \triangleq \{ w \sqsubseteq t \mid & w \in V^*,
N_t(w) = (a, b) \Longrightarrow
a, b \in (A\setminus V) \cup \{\Terminator\}\: \}
\end{array}
$$
We lift this construction from single traces to an entire system $\SysS$ and let: 
$$
V(S)\triangleq \bigcup_{t \in \Lang{\SysS}} V(t).
$$
\end{definition}

To conclude the definition of the neighbour representation, we let
\[\Pair{\SysS}{V} = \bigcup_{w \in V(S)} \Pair{\SysS}{w}.\]
to denote the set of neighbours of the maximal subtraces over $V(S)$.

To compare systems modulo unshared actions, it is convenient to treat each maximal subtrace $w$ over $V$ as a single abstract unit. Otherwise, the internal decomposition of such subtraces may interfere with neighbour extraction, especially for $k$-neighbour contexts. For this reason, we introduce a lifted trace representation in which maximal subtraces in $V(S)$ are handled as atomic tokens.

\begin{definition}[Trace Lifting via $V(\SysS)$]\label{def:trace-lifting}
Let $\SysS$ be an LTS and $V \subseteq A$. We define the \emph{lifted alphabet} as:
\[
A^\sharp \triangleq (A \setminus V) \cup V(S)
\]
Given a trace $t \in \Lang{\SysS}$, its \emph{lifted version} $t^\sharp \in (A^\sharp)^*$ is the unique sequence of tokens $x_1 x_2 \dots x_m$ such that their concatenation (interpreting each $x_j \in V(t)$ as a sequence over $A$) equals $t$, and for each $j \in \{1, \dots, m\}$, either $x_j \in (A \setminus V)$ or $x_j \in V(t)$.
Each maximal subtrace $w \in V(t)$ is treated as a single atomic element in $t^\sharp$.
\end{definition}

From this point onward, all traces are considered in their lifted form. Accordingly, neighbour contexts are computed over lifted traces $t^\sharp$, so that each maximal unshared subtrace $w$ occupies exactly one position in the local context.

We now introduce an abstraction mechanism for unshared actions while preserving their local structural role. We first define the abstraction on a single system, and then lift it to comparisons between systems with different alphabets.

\begin{definition}[Abstraction function]
\label{def:mu_v}
Let $\SysS$ be an LTS and let $V \subseteq A$ be a subset of the alphabet. Let $\Mask \notin A$ be a special symbol. We define the abstraction function:

\[
\Masking_V : A^\sharp \cup \{\Terminator\}  \longrightarrow (A\setminus V) \cup \{\Mask,\Terminator\}
\]
as
\[
\MaskFunExt{V}{x} \triangleq 
\begin{cases} 

\Mask & \text{if } x \in V(S),\\
x & \text{if } x \in ( A \setminus V) \cup \{\Terminator\}\\
\end{cases}
\]
\end{definition}

The function $\Masking_V$ maps every maximal unshared subtrace to the same abstract symbol $!$, while preserving the shared actions and the boundary symbol.

Next, we define the application of the function $\Masking_V$ to the neighbour contexts using the lifted traces.

\begin{definition}[Abstract neighbour contexts on lifted traces]
\label{df:renaming}
Let $\SysS$ be an LTS, with $V \subseteq A$ as its subset, and let $\Masking_V : A^\sharp \cup \{\Terminator\} \to (A \setminus V) \cup \{\Mask,\Terminator\}$ be the associated abstraction function. 
Given a trace $t \in \Lang{\SysS}$ and its lifted version $t^\sharp = x_1 \dots x_m \in (A^\sharp)^*$, 
we define the abstract neighbour of a token $x_j \in A^\sharp$ in $t^\sharp$ as
\[
\MaskFunExt{V}{\PairTraceExt{t^\sharp}{x_j}} \triangleq (\MaskFunExt{V}{x_{j-1}}, \MaskFunExt{V}{x_{j+1}})
\]
where $x_0 = \Terminator$ and $x_{m+1} = \Terminator$.

Accordingly, for any token $x \in A^\sharp$, we define its abstract neighbour set in the system $\SysS$ as:
\[
\MaskFunExt{V}{\Pair{\SysS}{x}} 
 \triangleq \bigcup_{t \in \Lang{\SysS}} \{\, \MaskFunExt{V}{\PairTraceExt{t^\sharp}{x_j}} \mid 1 \le j \le m \text{ and } x_j = x \,\}.
\]

Finally, we define the \textit{global abstract neighbour set} 
of the unshared part of the system as
\[
\MaskFunExt{V}{\Pair{\SysS}{V}} \triangleq 
\bigcup_{w \in V(S)} 
\MaskFunExt{V}{\Pair{\SysS}{w}}.
\]
\end{definition}

We extend the same construction to $k$-neighbour contexts over lifted traces.

\begin{definition}[Abstract $k$-neighbour contexts on lifted traces]
\label{def:absneighcontlt}
Let $t^\sharp = x_1x_2\ldots x_m \in (A^\sharp)^*$ be a lifted trace. For $k \ge 1$, the abstract $k$-neighbour of a token $x_j$ in $t^\sharp$ is defined as
\[
\MaskFunExt{V}{\KPairTraceExt{t^\sharp}{x_j}} \triangleq (\MaskFunExt{V}{x_{j-k}},\dots,\MaskFunExt{V}{x_{j-1}}, \MaskFunExt{V}{x_{j+1}},\dots,\MaskFunExt{V}{x_{j+k}})
\]
where out-of-bound indices are interpreted as the boundary symbol $\Terminator$.
	
Accordingly, the abstract $k$-neighbour set of a token $x \in A^\sharp$ in the system $S$ is obtained by collecting all such abstract $k$-neighbours over the lifted traces of $S$.
\end{definition}

By computing $k$-neighbour contexts on lifted traces rather than on concrete traces, each maximal unshared subtrace occupies exactly one position in the $k$-neighbour contexts. This prevents the abstraction from leaking the exact internal length of unshared sequences and ensures that the comparison remains stable under sequence refinement.

Having defined the abstraction for a single system, we can now lift it to comparisons between two systems, $S_1$ and $S_2$, with different observable alphabets. By taking $V = A_1 \triangle A_2$, the abstraction masks their unshared actions while preserving the local contexts induced by the shared ones.



To conclude this section, we define neighbour equivalence modulo $\Masking_V$. This notion allows us to compare two systems with different observable alphabets by abstracting away their unshared actions and focusing on the local contexts that remain structurally comparable.

\begin{definition}[Neighbour equivalence modulo $\Masking_V$]
\label{def:mu-neighbour-eq}

Let $\SysA$ and $\SysB$ be two LTSs with observable alphabets $A_1$ and $A_2$. Let $V = A_1 \triangle A_2$ be the set of unshared actions, and let $\Masking_V$ be the corresponding abstraction function defined over the lifted alphabets of $S_1$ and $S_2$.

We say that $\SysA$ and $\SysB$ are \emph{neighbour equivalent modulo $\Masking_V$}, written $\SysA \sim_{\Masking_V} \SysB$, if the following conditions hold:
\begin{enumerate}
    \item \textbf{Equivalence of abstract neighbour contexts for unshared actions:}\label{cond:eq-not-shared}

    \[
    \MaskFunExt{V}{\Pair{\SysA}{V}} = \MaskFunExt{V}{\Pair{\SysB}{V}}
    \]

    \item \textbf{Equivalence of abstract neighbour contexts for shared actions:}\label{cond:eq-shared}
    \[
    \forall \GenericAct \in A_1 \cap A_2, \;
    \MaskFunExt{V}{\Pair{\SysA}{\GenericAct}} = \MaskFunExt{V}{\Pair{\SysB}{\GenericAct}}
    \]
\end{enumerate}
\end{definition}

In operational terms, this definition requires two complementary forms of compatibility. First, the unshared parts of the two alphabets must induce the same abstract boundary contexts. Second, the shared actions must preserve the same local abstract contexts once the unshared fragments have been masked. Together, these two conditions allow the comparison of systems with different observable alphabets while ignoring differences internal to the unshared parts.

\begin{example}\label{ex:sequence_equivalence}
To illustrate how the framework works in practice, consider the sequence evolution scenario described by the two systems $S_0$ and $S_3$ shown in Figure~\ref{fig:evolution_system}. In this case, the action $c$ is refined into the sequence of operations $ef$.
Let $S_0$ and $S_3$ be two systems with traces $\Lang{S_0} = \{bac,bad\}$ and $\Lang{S_3} = \{baef, bad\}$.
The alphabets are $A_0 = \{a,b,c,d\}$ and $A_3 = \{a,b,d,e,f\}$.
Hence, the unshared and shared actions are
\[
V = A_0 \triangle A_3 = \{c, e, f\}, \qquad A_0 \cap A_3 = \{a,b,d\}.
\]

Under trace lifting, the maximal unshared subtrace $c$ in $S_0$ and the maximal unshared subtrace $ef$ in $S_3$ are each treated as a single token. 

\paragraph{Condition~\ref{cond:eq-not-shared}.} From Definition~\ref{def:variant_sequence} we obtain:
\[
V(S_0) = \{c\}, \qquad V(S_3) = \{ef\}.
\]
Their neighbour sets are
\[
\Pair{S_0}{c} = \{(a,\Terminator)\} \quad \text{and} \quad \Pair{S_3}{ef} = \{(a, \Terminator)\}.
\]
Therefore,
\[
\Masking_V(\Pair{S_0}{V})=\Masking_V(\Pair{S_3}{V})=\{(a, \Terminator)\},
\]
so the first condition is satisfied.
    
\paragraph{Condition~\ref{cond:eq-shared}.} We now consider the shared actions $a$, $b$ and $d$.
For action $a$, the neighbour sets differ before abstraction:
\[
\Pair{S_0}{a} = \{(b, c),(b,d)\} \qquad
\Pair{S_3}{a} = \{(b, ef),(b,d)\}.
\]
However, after abstraction, we obtain
\[
\MaskFunExt{V}{\Pair{S_0}{a}} = \{(b, \Mask),(b,d)\} = \MaskFunExt{V}{\Pair{S_3}{a}}.
\]
For action $b$, the neighbour sets already coincide: $\Pair{S_0}{b} = \{ (\Terminator,a)\}$ and $\Pair{S_3}{b} = \{(\Terminator,a)\}$, hence
\[
\MaskFunExt{V}{\Pair{S_0}{b}} = \{(\Terminator,a)\} = \MaskFunExt{V}{\Pair{S_3}{b}}.
\]
For action $d$, we similarly have $\Pair{S_0}{d} = \{ (a, \Terminator)\}$ and $\Pair{S_3}{d} = \{(a,\Terminator)\}$, and therefore
\[
\MaskFunExt{V}{\Pair{S_0}{d}} = \{(a,\Terminator)\} = \MaskFunExt{V}{\Pair{S_3}{d}}.
\]
Since both conditions hold, we conclude that $S_0 \sim_{\Masking_V} S_3$.
\end{example}

This abstract comparison covers the evolution scenarios introduced in Section~\ref{sec:methodology}, including renaming, splitting, fusing, and more general many-to-many correspondences. 
If two systems are neighbour equivalent modulo $\Masking_V$, then their alphabet differences are locally functionally similar under the chosen abstraction. This supports the hypothesis that when a system refines its behaviour while maintaining the same functionality, some refined actions of the alphabets may correspond to the same monitoring-relevant role. This information can guide the adaptation of monitoring artefacts needed to preserve monitoring continuity.

\section{Monitor and Specification Adaptation}
\label{sec:spec-updat}

The framework introduced in the previous section lays the foundation for adapting runtime monitors to evolved systems. Rather than attempting a rigid, fully automated translation, which is often unfeasible for arbitrary alphabet changes, it identifies correspondences between actions in different observable alphabets by comparing their local contexts. In this section, we discuss how these neighbour equivalent actions can support the adaptation of monitors and specifications.

A useful distinction must be made between the \emph{intent of the property} and the \emph{observable vocabulary} through which that intent is monitored. In the scenarios considered here, the requirement itself is assumed to remain stable, while the set of observable actions may change due to renaming, refinement, abstraction, or loss of visibility. The purpose of adaptation is to systematically reconstruct a suitable observational interpretation of that property over the evolved alphabet.

Before detailing the adaptation patterns, it is important to emphasise that this mapping framework is designed to be as independent as possible from the specific formal language used for the specifications. While we employ standard LTL syntax to formalise the substitution rules, the underlying concepts represent general, language-agnostic adaptation strategies. 

While translating arbitrary LTL specifications across different alphabets poses known theoretical challenges, our neighbour-based compatibility provides valuable local evidence that newly observed actions preserve the same \emph{monitoring-relevant role} as previously observed ones. Accordingly, the adaptation patterns discussed below serve as practical strategies for evolving temporal specifications.

Let $\varphi$ be an LTL specification over the alphabet $A_1$ of a system $\SysA$, and let $\SysB$ be the evolved version of $\SysA$ with the modified alphabet $A_2$. Assume that $\SysA \sim_{\Masking_V} \SysB$, which indicates local compatibility between the changed actions of $\SysB$ and the missing actions of $\SysA$ under the chosen abstraction.
The goal is to determine how to adapt $\varphi$ so that monitoring remains meaningful for $\SysB$.
Furthermore, we initially define these patterns focusing on single-action substitutions (e.g., renaming, splitting) to isolate the core logic of each transformation. However, as we will demonstrate, these foundational rules can be naturally extended and generalised to accommodate sequences of actions, covering cases where atomic operations are refined into more complex executions.

\begin{definition}[Renaming adaptation]
Let $a \in A_1$ and $b \in A_2$, such that $V=\{a,b\}$. Let $\SysA \sim_{\Masking_V} \SysB$, then the \emph{renaming adaptation} of $\varphi$ is the formula
\[
\varphi[a \mapsto b].
\]
\end{definition}

Renaming is the most straightforward adaptation pattern supported by the framework. If the correspondence between $a$ and $b$ is correct and the alphabet evolution is observationally a pure renaming, replacing $a$ with $b$ perfectly preserves the intended monitoring role of the action.

\begin{definition}[Splitting adaptation]
Let $a \in A_1$ and let $U=\{u_1,\dots,u_m\}\subseteq A_2$ be a set of actions that substitutes $a$. The \emph{splitting adaptation} of $\varphi$ is the formula
\[
\varphi[a \mapsto (u_1 \vee \dots \vee u_m)].
\]
\end{definition}
The splitting adaptation provides a robust heuristic when the original action $a$ is monitored as a witness of a system step, and the actions in $U$ represent alternative observable realisations of that same step. Under this interpretation, the disjunction $u_1 \vee \dots \vee u_m$ successfully preserves the idea that any of the new actions may satisfy the role previously played by $a$. While its soundness depends on the specific occurrences of $a$ in the formula, it provides a direct, intuitive adaptation path.

\begin{definition}[Sequence adaptation]
Let $a \in A_1$ and let $w \in V(\SysB)$ be the subtrace identified as a substitute for $a$. The \emph{sequence adaptation} of $\varphi$ expands $a$ into a temporal sequence using the LTL Next operator ($\bigcirc$), yielding the formula:
\[
\varphi[a \mapsto (u_1 \wedge \bigcirc(u_2 \wedge \dots \wedge \bigcirc u_m)\dots)].
\]
\end{definition}
The sequence adaptation bridges the gap between different levels of abstraction. While splitting captures alternative single actions, sequence adaptation handles cases where an atomic system step is refined into a concrete sequence of operations. It is important to note that this substitution is not strictly sound for arbitrary LTL formulas. Rather, this rule should be interpreted as a restricted-fragment adaptation pattern, appropriate when the original action acts as a witness of an abstract step. By treating the sequence in this manner, the pattern directly leverages the neighbour analysis introduced in Section~\ref{sec:neighbour-equiv}, ensuring that the expanded temporal footprint in the monitor matches the abstract boundaries of the original action.

\begin{definition}[Fusing adaptation]
Let $U = \{a_1, \dots, a_m\} \subseteq A_1$ be a set of actions that are replaced by a single action $b \in A_2$. Let $S_1 \sim_{\mu_V} S_2$ with $V = U \cup \{b\}$, and assume that for all $a_i \in U$, $\Pair{\SysA}{a_i}$  are pairwise disjoint.

Under the additional assumption that the local contexts of the actions in $U$ are pairwise disjoint, the following contextualised replacement provides a plausible heuristic for distinguishing fused actions.
The adaptation corresponds to the formula:
\[
\varphi [ a_i \mapsto( \bigvee_{(x,y) \in N^{S_1}(a_i)}  x \wedge \bigcirc (b \wedge \bigcirc y) \ )]
\]
for each $i \in \{1, \dots, m\}$.
\end{definition}
The fusing operation combines multiple previously distinct actions into a single observable action, potentially leading to a loss of observability.
The fusing adaptation leverages the local trace structure to determine which of the original actions is currently occurring. If the actions in $U$ originally occurred in mutually disjoint contexts, conditioning the observation of $b$ on its neighbour context prevents the monitor from issuing false positives. However, this rule should be interpreted conservatively: it is not a general sound translation for arbitrary LTL formulas, but rather a context-sensitive heuristic designed to distinguish fused actions based on their local context.

\begin{remark}[Multi-mapping and Compositional Adaptation]
As established in Section~\ref{sec:neighbour-equiv}, the neighbour equivalence framework natively supports general $M$-to-$N$ (multi-mapping) correspondences by abstracting maximal unshared subtraces into atomic tokens. The framework seamlessly identifies structural compatibilities regardless of the length of the involved sequences.

However, moving from structural equivalence to explicit LTL specification adaptation introduces semantic challenges. While the patterns defined above provide direct substitution rules for foundational configurations (i.e., $1$-to-$1$, $1$-to-$N$, and $N$-to-$1$), creating a single, general LTL syntactic rule for arbitrary $M$-to-$N$ mappings is highly non-trivial. 

Instead, complex adaptations can often be constructed compositionally. For instance, one may conceptually apply a fusing adaptation to the $M$ original actions to obtain an intermediate abstract token, and then apply a sequence or splitting adaptation to relate it to the $N$ new actions. In cases where direct syntactic composition is not straightforward, the neighbour-based compatibility still fulfils its primary role: it narrows the space of plausible correspondences, guiding a structured and evidence-based manual rewriting of the specification.
\end{remark}

In conclusion, neighbour equivalence modulo $\Masking_V$ supports monitor and specification adaptation by guiding the runtime discovery of candidate evolutionary mappings.
It is most naturally applicable to direct translation scenarios, such as renaming, splitting, and selected cases of sequence expansion, while serving as a precise diagnostic guide for complex transformations, such as fusing and multi-mapping. By narrowing the search space of candidate correspondences, it helps transform an otherwise unguided manual update into a more systematic, evidence-based process.

\section{Conclusions and Future Work}
\label{sec:conclusion}

This paper has addressed the problem of runtime semantic discovery in evolving systems, namely the identification of monitoring-relevant correspondences between actions when the observable alphabet changes across system versions. This problem arises naturally in runtime verification when updates, reconfigurations, or architectural changes preserve the high-level intent of a system while altering the actions through which that intent is observed.

Our approach is based on the idea that the role of an action can often be approximated from the local context in which it occurs. Starting from this intuition, we introduced neighbour-based representations of traces, defined neighbour equivalence and its k-neighbour generalisation, and showed how these notions can be used to compare systems with equal or partially different alphabets through abstraction. The resulting framework is intentionally lightweight and runtime-oriented: rather than reconstructing a global behavioural model, it extracts local structural evidence from observed traces.

The paper also clarifies the limits of this approach. Neighbour equivalence is a useful structural invariant, but it is strictly weaker than trace equivalence. For this reason, it should not be interpreted as a complete characterisation of behavioural preservation. Instead, it provides a basis for deriving candidate semantic correspondences between actions, which can then support monitor and specification adaptation in selected evolution scenarios. In particular, the discussion in Section \ref{sec:spec-updat} suggests that the framework is most naturally applicable to cases close to renaming and controlled forms of splitting, while more complex transformations such as fusing generally require additional semantic information beyond local neighbour preservation.


Broadly speaking, the contributions of the paper span both theoretical and practical dimensions. On the theoretical side, it isolates a hierarchy of local equivalences that are computable on traces and related in a precise way to stronger behavioural notions such as trace equivalence. On the practical side, it provides a runtime-friendly criterion for narrowing the space of possible monitor and specification adaptations when the monitored system evolves and previously known action vocabularies become partially obsolete.
More broadly, we believe that local trace-based discovery of candidate semantic correspondences is a promising step toward maintaining monitoring continuity in software systems whose observable behaviour evolves over time.

Several directions remain open for future work. To enhance the discriminative power of our technique, a first extension involves enriching the neighbour-based abstraction with quantitative information, such as frequencies and probabilistic profiles. A second direction is to characterise more precisely the logical fragments for which specification adaptation rules can be proved sound. A third direction is an empirical evaluation of the framework, assessing both its practical cost and the accuracy and robustness of neighbour-based correspondences under noisy, partial, or ambiguous traces. As a first step in this direction, we have begun to implement the framework.

\bibliographystyle{ACM-Reference-Format}
\bibliography{references}

\end{document}